\documentclass[article,amsthm,amsmath,amssymb]{lipics-v2021-modified}

\usepackage[]{graphicx}
\usepackage{amssymb}
\usepackage{enumitem}
\usepackage{hyperref}
\usepackage{amsmath}
\usepackage{amsthm}
\usepackage{multicol}
\usepackage{qcircuit}
\usepackage{braket}

\newcommand{\minrank}{\mathrm{minrank}}
\newcommand{\cz}{\mathrm{CZ}}
\newcommand{\cx}{\mathrm{CX}}
\newcommand{\xcx}{\mathrm{XCX}}

\title{Harnessing the Power of Long-Range Entanglement for Clifford Circuit Synthesis} 

\author{Willers Yang}{IBM Quantum, MIT-IBM Watson AI Lab, Cambridge, Massachusetts 02142, USA}{willers.yang@ibm.com}{}{}

\author{Patrick Rall}{IBM Quantum, MIT-IBM Watson AI Lab, Cambridge, Massachusetts 02142, USA}{patrickjrall@ibm.com}{}{}

\authorrunning{W. Yang and P. Rall} 

\Copyright{IBM Quantum} 



\category{} 

\relatedversion{} 

\acknowledgements{We thank Sergei Bravyi, Shelly Garion, Alexander Ivrii, Dmitri Maslov, and Derek Wang for helpful discussions. This study was funded by IBM Quantum.}

\nolinenumbers 

\ArticleNo{1}

\begin{document}

\maketitle

\begin{abstract}
In superconducting architectures, limited connectivity remains a significant challenge for the synthesis and compilation of quantum circuits. We consider models of entanglement-assisted computation where long-range operations are achieved through injections of large GHZ states. These are prepared using ancillary qubits acting as an “entanglement bus,” unlocking global operation primitives such as multi-qubit Pauli rotations and fan out gates. We derive bounds on the circuit size for several well-studied problems, such as CZ circuit, CX circuit, and Clifford circuit synthesis. In particular, in an architecture using one such entanglement bus, we give an $O(n^3)$-complexity synthesis scheme for arbitrary Clifford operations requiring at most $2n+1$ layers of entangled-state-injections. In a square-lattice architecture with two entanglement buses, we show that a graph state can be synthesized using at most $\lceil \frac{1}{2}n\rceil +1$ layers of GHZ state injections, and Clifford operations require only $\lceil\frac{3}{2} n \rceil+ O(\sqrt n)$ layers of GHZ state injections. 
\end{abstract}

\section{Introduction}

Unlike classical random access memories where direct access to arbitrary bits comes at a low cost, quantum operations across non-adjacent qubits often incur significant additional overhead.  One common resolution is to use SWAP gates to bring qubits to adjacent positions. However, when we use SWAP gates for qubit routing, we may suffer an overhead in depth that is linear in the number of qubits. While this is not concerning for exponential speedups in theory, in practice, this overhead could render quantum algorithms with only mild polynomial speedups useless and otherwise dull the quantum competitive edge. 

Another solution is to use long-range entanglement to implement non-local operations. Local measurements with feed-forward corrections allow us to prepare quantum states with long-range entanglement in constant depth. Using gate teleportation and similar techniques, these states can be used as a resource to implement long-range two-qubit gates and even global $n$-qubit gates. This observation suggests the following trade: sacrifice a constant fraction of the qubits to act as an ``entanglement bus'' and obtain a certain flavor of all-to-all connectivity in exchange. 

These techniques are widely considered in surface code architectures, especially lattice-surgery \cite{1808.02892}, which reformulates Clifford + T circuits in terms of ancilla-assisted multi-qubit Pauli rotations. Works on surface code routing leverage constant-depth preparation of Bell states to facilitate long-range CNOT gates \cite{2110.11493, 2204.04185}. Other models leverage Hamiltonian time evolution to implement certain $n$-qubit gates and discuss their utility towards implementing permutations \cite{2206.01766} and Clifford operations \cite{1707.06356, 2012.09061, 2207.08691}.  It is also well known that certain families of interesting quantum states in physics are easy to prepare using measurement and feedback \cite{2103.13367, 2210.17548}. Using entangled states as a resource for computation is a central idea in the field of measurement-based quantum computation \cite{quant-ph/0508124}, whose techniques enable us to trade circuit depth for circuit width.

Previous works on surface code compilation have either performed numerical studies on the speed of implementing fixed sequences of CNOT gates \cite{2110.11493}, or on asymptotic bounds on the implementation of permutations using entanglement-routing \cite{2204.04185}. Our work takes inspiration from these proposals while also exploiting the particular structure of the Clifford group and studying the leading coefficient in the synthesis performance. This approach allows us to incorporate more sophisticated optimizations, achieve highly parallelized circuits, and have competitive upper bound guarantees using a smaller fraction of ancillary qubits. 

Efficient circuit synthesis of Clifford operations is not just central to the implementation of fault-tolerant quantum algorithms, but it is also extensively studied in various other models. In particular, using single-qubit gates and two-qubit entangling gates such as the CNOT and CZ gates, we can show that arbitrary Clifford circuits can be synthesized using at most $7n-4$ layers of two-qubit gates with linear-nearest-neighbor (LNN) architecture \cite{2210.16195}, and at most $2n+O(\log^2(n))$ layers of two-qubit gates with all-to-all connectivity \cite{2201.05215}. When we allow global operations, an algorithm exists that computes the optimal decomposition of a Clifford operation into Pauli rotations, achieving a worst-case gate count of $2n+1$ \cite{2102.11380}. Finally, in architectures that allow a more powerful global tunable gate, Cifford operations may require only constant depth \cite{2207.08691}. 






The remainder of the paper is organized as follows. First, we will describe the GHZ bus models in detail in Section~\ref{sec:Model} and relate them to other proposals in prior work. Then, in Section~\ref{sec:LayeredSynth}, we will present several optimization techniques to efficiently synthesize various classes of Clifford circuits, starting with CZ circuits, CX circuits, and Hadamard-free circuits in Subsections \ref{ssec:CZ}, \ref{ssec:CX}, and \ref{ssec:Hfree} respectively. In particular, our construction for CZ circuits achieving depth $\lceil \frac{1}{2} n \rceil +1$ in a model with two GHZ buses can also be applied to the synthesis of graph states. Then, combining these optimization techniques, we arrive at our main results on Clifford synthesis in Subsection~\ref{ssec:Cliff}: firstly, in a model with one GHZ bus, we present a simplified construction achieving the optimal GHZ injection depth guarantee of $2n+1$ by \cite{2102.11380}; and secondly, in a more powerful model with two GHZ buses with square lattice connectivity, we present a highly parallelized construction achieving  GHZ injection depth $\lceil \frac32\rceil n+O(\sqrt{n})$ for Clifford synthesis. Lastly, we present some lower bounds in Appendix~\ref{app:lower_bounds}, and additional derivations in Appendix~\ref{app:derivations_examples}. Our results are summarized in Table~\ref{tab:comparison}.

\begin{table}
\begin{center}
{\color{blue}}
\begin{tabular}{p{2.3cm}||p{2.0cm}|p{4.1cm}|p{3.8cm}}\hline\hline
           Model           &  Depth Metric       &State Synthesis                         & Clifford Synthesis  \\\hline
         & & --    & $\geq 2n+1$ \cite{quant-ph/0701194}    \\
 {LNN}   & CNOT  & $\leq2n+2$ \cite{1808.02892}\footnote{Note this construction also reverses the input qubits.}    & $\leq 7n-4$ \cite{2210.16195}  \\\hline
 {All-to-all} & CNOT & $\leq\frac{n}{2} +O(\log^2(n)) $\cite{2201.05215}  & $\leq 2n+ O(\log^2(n))$ \cite{2201.05215} \\\hline
            &        &$\geq$minrank$(G)$ (clique flips)  & $\geq 0.648n - 2$ (Prop. \ref{prop:lowerbound}) \\
 {Linear GHZ Bus} &GHZ injection  &$\leq$minrank$(G)+1$ (Prop. \ref{prop:CZ(minrank)})  &  $\leq 2n +1$ (Corr. \ref{cor:cliff(2n)},\cite{2102.11380})  \\\hline
 {Dual Snake}  &GHZ injection & $\leq\lceil \frac{n}{2}\rceil + 1$ (Prop. \ref{prop:CZ(n/2)}) & $\leq \lceil \frac32\rceil n +O(\sqrt{n})$ (Corr. \ref{cor:cliff(1.5n)})\\\hline
\end{tabular}

\end{center}
\caption{{\color{blue}}\label{tab:comparison}Comparison of best-known upper and lower bounds on circuit depth for graph states and Clifford operations. The linear GHZ bus and dual snake models are proposed in this work and explained in Section~\ref{sec:Model}.}
\end{table}

\section{Models}\label{sec:Model}

Our results are chiefly inspired by surface code architectures, in which the allocation of ``entanglement bus'' qubits to facilitate long-range interactions is common in several works \cite{1808.02892, 2110.11493}. Rather than studying the capabilities of a complex architecture for large quantum circuits in practice, we design simplified models to capture \emph{only} the impact of GHZ state injection on an architecture with otherwise poor connectivity. We expect improved connectivity to be the primary benefit of GHZ state injection, and this architecture lets us quantify the improvement. In this section, we describe and discuss the capabilities of the model.

\begin{figure}[!htb]
\centering
  \includegraphics[width=\linewidth]{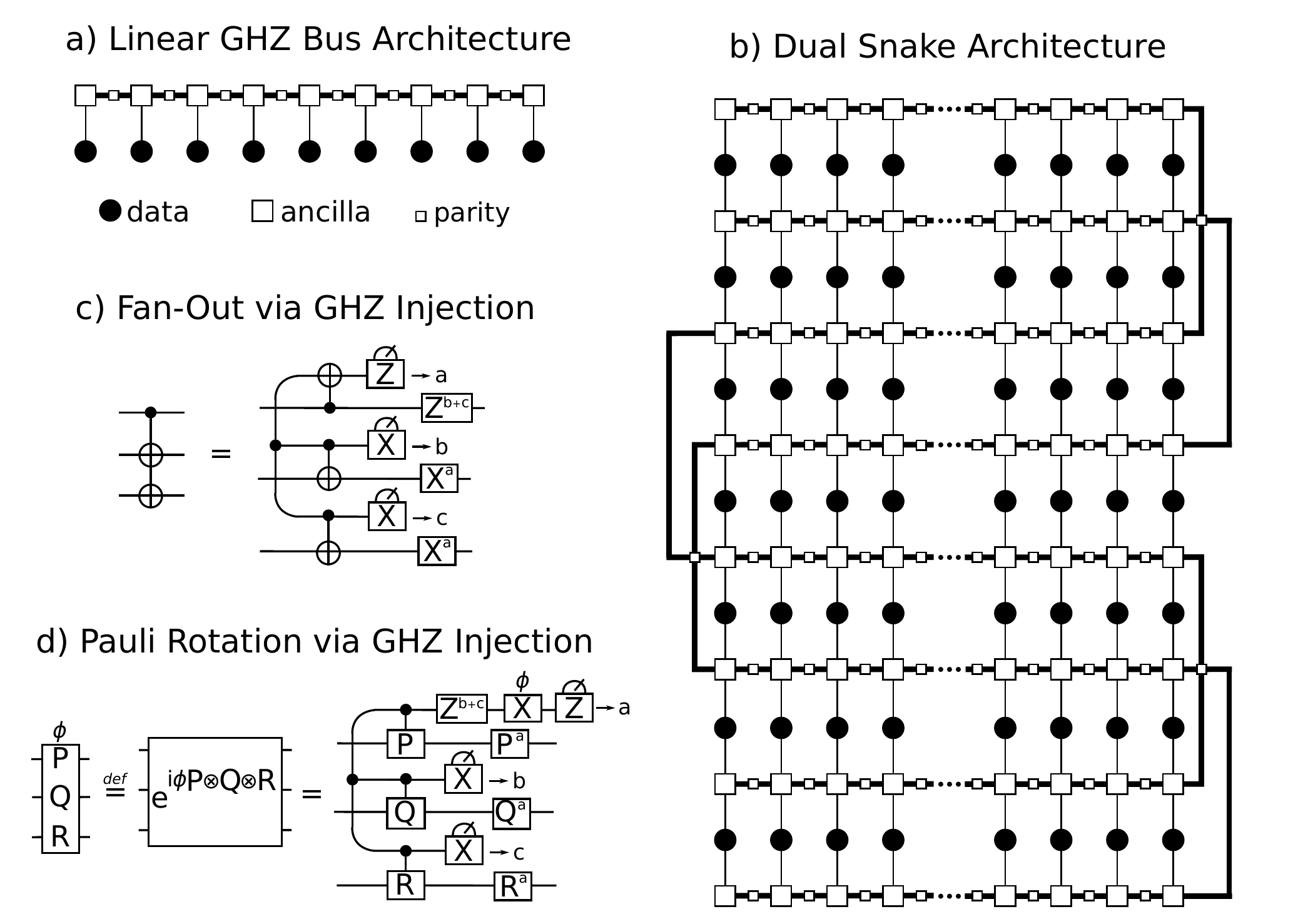}
  \caption{Architectures and primitive operations considered in this manuscript. a) A Linear Nearest Neighbor (LNN) architecture with a `GHZ bus': a `rail' of ancilla qubits reserved for the preparation of GHZ states (see Figure~\ref{fig:GHZ_prep}). b) A `dual snake' architecture compatible with a square lattice of qubits featuring two intertwining GHZ buses. In the limit of many qubits, only about half of the chip area is dedicated to ancillae. This architecture permits the parallelization of two layers of primitive operations provided they act on disjoint sets of qubits. c) Implementation of a fan-out gate using a GHZ state prepared on GHZ bus. d) Implementation of a Pauli rotation gate via injection of a GHZ state. }\label{fig:LNN_GHZ_Bus}
\end{figure}

\begin{definition} \emph{GHZ Bus Architectures} enable a set of gates acting on $n$ qubits. With the qubits enumerated $1\ldots n$, the operations are:
\begin{itemize}
 \item all single qubit gates,
 \item and $k$-qubit gates acting on $k$ adjacent qubits $i,...,i+(k-1)$ from the following families:
 \begin{itemize}
    \item CNOT fan-out: if a control qubit is $\ket{1}$, apply $X$ to any subset of the other $k$ qubits.
    \item Pauli rotation: for any phase angle $\phi$ and any multi-qubit Pauli matrix $P$ supported on the $k$ qubits, apply the unitary $\exp(i\phi P)$.
 \end{itemize}
\end{itemize}
If any operations act on non-overlapping ranges of qubits $l_1...r_1$ and $l_2...r_2$ such that $r_1<l_2$, then the gates can be performed simultaneously. A layer of parallel gates consisting entirely of the latter class of $k$-qubit gates is called a GHZ state injection layer.  Single qubit gates are considered instant/free. A linear GHZ bus architecture allows one such layer of gates at a time, while a dual snake architecture parallelizes two such layers, provided they act on disjoint sets of qubits. See Figure~\ref{fig:LNN_GHZ_Bus}.
\end{definition}

The $k$-qubit fan-out and Pauli measurement gate families are implemented via a single GHZ state preparation and injection, as shown in Figures~\ref{fig:LNN_GHZ_Bus}~a) and~b). We believe these families of operations equivalently capture the power of GHZ state injection, although their interconversion demands a different connectivity model and additional measurement feedback (see Figure~\ref{fig:interconversion} in Appendix~\ref{app:derivations_examples}). 

The relative cost of the nearest-neighbor two-qubit gates and the GHZ state injections depends on the implementation of the model. Recall that in a surface code architecture based on lattice surgery, CNOT gates are not native and require an ancilla qubit in order to implement. Furthermore, the synthesis of large ancilla patches containing GHZ states can be performed simultaneously as the preparation of a CNOT ancilla. Thus, CNOT and GHZ state injection layers have the same cost! A more general layer of nearest-neighbor interactions, such as a layer of SWAP gates, may need as much time as three GHZ state injection layers! The situation is reversed in a model in which some two-qubit gates are native. In such an architecture, GHZ state preparation (Figure~\ref{fig:GHZ_prep}) and subsequent injection require three CNOT layers and additional measurement feedback. A similar argument applies to the parity check qubits: in a surface code architecture the parity checks can be performed with no additional space cost, but this is not the case in a near-term architecture.


We briefly compare our approach to some other works. First, \cite{2204.04185} consider a broader family of connectivity graphs than LNN, but instead of GHZ states only focus on Bell state-enabled long-range swaps, and their impact on implementing permutation circuits. Our interest in LNN specifically is that it serves as a stepping stone toward square-lattice architectures that likely capture surface code constructions' capabilities on superconducting hardware. While the study of permutation circuits is a natural approach for quantifying the power of ancilla-enabled long-range gates, we find Clifford circuits enable a richer family of optimizations. 

Second, \cite{2110.11493} investigates the performance of a square lattice architecture in which each data qubit is padded with three additional ancilla qubits for routing. In comparison, our dual snake architecture merely adds one ancilla per data qubit. Is the dual snake architecture slower for implementing Clifford gates than a layout with more ancillae? A theoretical analysis of the performance of the parallel CNOT routing considered by \cite{2110.11493} yields that $O(n^{1.5})$ layers suffice: $n$ CNOT layers suffice to implement a Clifford gate, and each layer requires $\Theta(\sqrt{n})$ operations. This bound is much looser than the $2n+1$ synthesis bound \cite{2012.09061} with just a single entanglement bus, as well as the one we derive for our dual snake scheme which achieves $\lceil\frac{3}{2}n\rceil + O(\sqrt{n})$, despite both needing fewer ancillae.  We do not know of a method for leveraging the additional routing ancilla qubits considered by \cite{2110.11493} to increase performance, so perhaps they are not necessary. We note that our dual snake architecture has GHZ buses that cross each other. This is permitted since \cite{2110.11493} shows that two Bell states can still be prepared simultaneously in such a layout.

Third, we note that the architecture admits a `clique flip' operation (defined in Figure~\ref{fig:cliffEg}~a)), which is equivalent to applying CZ on all pairs of the $k$ qubits it acts on. The name of this operation is motivated by CZ circuit synthesis: CZ circuits are equivalent to graphs, and the clique flip operation lets us toggle all the edges of the graph within a clique of our choosing. This operation appeared in \cite{1707.06356, 2012.09061} as a special case of the global Molmer-Sorensen (GMS) gate. Since clique flip operations are local-Clifford-equivalent to Pauli rotations $\exp(i \frac{\pi}{4} P)$, we find that the constructions from this line of work already capture some, but not all, of the power of the model we consider. Indeed, other than the clique flip operation, GMS gates and GHZ injections seem to have rather different capabilities and resource requirements. On the one hand, recent work \cite{2206.01766} shows that general GMS gates can implement Clifford gates in constant depth. On the other hand, GMS gates are inspired by Hamiltonian evolution on hardware with all-to-all connectivity (like ion trap quantum computers). Our GHZ bus model is more inspired by the limited connectivity of superconducting quantum computers running surface codes. While GHZ state injection can implement clique flips, a special case of the GMS gate, it is not clear how to use GHZ injection to implement general GMS gates. Similarly, while a unitary circuit with two Clique flips suffices to implement fan-out, it is unclear how to perform CNOT fan-out with just one clique flip.

Fourth, it is a well-known result in the theory of measurement-based quantum computation that Clifford gates can be implemented in constant depth on photonic hardware \cite{quant-ph/0508124}. This is achieved by rendering the Clifford circuit into a sequence of gate teleportations, causing the overall \emph{width} of the circuit to scale with the circuit complexity instead. We are interested in superconducting architectures where the width of the circuit is fixed. 

Finally, we briefly discuss the feasibility of implementing this model in current-generation IBM hardware. Broadly, the hardware seems to have the necessary capabilities: mid-circuit measurement and feed-forward correction via \emph{dynamic circuits}, as well as connectivity that enables a limited version of the linear GHZ bus model. The GHZ injection circuits present the opportunity for significant savings in circuit depth which may be useful when coherence time is a limitation. However, the mid-circuit measurements they require also introduce a lot of new noise, and qubits reserved for GHZ states cannot store data. Under what circumstances are these sacrifices worth the improvement in depth?

\section{Depth-optimized Clifford Synthesis using GHZ-states}\label{sec:LayeredSynth}
{{\color{blue}} 

First, recall that up to a layer of single qubit Pauli gates, the group of Clifford operations on $n$ qubits is isomorphic to the $2n\times 2n$ binary symplectic group $Sp(2n,\mathbb{F}_2)$. The task of synthesizing Clifford operations using a restricted set of gates is represented by the diagonalization of a binary symplectic matrix using operations corresponding to the gate set. For example, when we have only two-qubit entangling gates $\{\cx, \cz\}$, Clifford operations can be decomposed into a layered computation in the form -L-CX-CZ-H-CZ-L- \cite{2003.09412, 1808.02892}, where -L- denotes a layer of single-qubit Clifford gates, -CX- and -CZ- denotes layers of circuits consisting entirely of $\cx$ and $\cz$ gates, and -H- denotes a layer of Hadamard gates applied to all qubits. 

Some Clifford operations can be synthesized with GHZ state injections with simpler circuits. For example, both fan out gates and clique flip operations can be implemented using only one GHZ state injection, while they require a circuit of depth $\Omega(n)$ in LNN when only two-qubit gates are available. In general, Phalla et al. gave an algorithm that, with some exceptions, decomposes a Clifford operation into a minimal number of Pauli rotations $\exp(i \frac{\pi}{4} P)$ \cite{2102.11380}. Their decomposition achieves a depth of $\leq 2n+1$ that can be implemented naturally within our GHZ bus model. However, it is not often easy to obtain this decomposition since a subroutine it relies on--the triangularization of binary matrices by congruence--fails in some exceptional cases \cite{botha}. Furthermore, while \cite{botha} does not give an explicit algorithm for obtaining these decompositions, we found that an algorithm based on their work requires $O(n^4)$ time. Lastly, though the decomposition by \cite{2102.11380} is guaranteed to use the minimal \textit{number} of global operations, it does not leverage our models' additional powers, such as the ability to parallelize multiple GHZ state injections.

In the remainder of this section, we will present constructive propositions for synthesizing various Clifford circuits, such as CZ circuits, CX circuits, and Hadamard-free circuits using GHZ state injections. Together, they lead to the two main results: firstly, in the linear GHZ bus model shown in Figure~\ref{fig:LNN_GHZ_Bus}~a), we show that any Clifford operation can be decomposed into a circuit with at most $2n+1$ GHZ state injection layers using $O(n^3)$ classical computation time. Secondly, we show that square lattice connectivity supporting the dual snake architecture in Figure~\ref{fig:LNN_GHZ_Bus}~b) can do so in depth $\leq \lceil \frac{3}{2}n\rceil +O(\sqrt{n})$.  Additionally, since graph states can be prepared using a CZ circuit, our constructions can also be extended to graph state synthesis. Unless otherwise specified, we will often use depth to refer to a circuit's GHZ state injection depth; that is, the number of parallel GHZ state injection stages.}

\subsection{-CZ- Transformations and Graph State Synthesis} \label{ssec:CZ}

We first establish some results on synthesizing -CZ- layers.  The central idea underpinning these methods is that -CZ- layers are equivalent to graphs since CZ gates are symmetric, self-inverse, and mutually commuting.  A particularly convenient $k$-qubit gate for manipulating these graphs is the `clique flip' operation defined in Figure~\ref{fig:cliffEg}~a), which implements a CZ gate on all pairs of qubits involved. In the graph representation, this operation toggles all edges in a clique. We defer the proofs of these results to Appendix~\ref{app:CZsynth}.

{{\color{blue}}
By relating -CZ- circuits to graphs and their adjacency matrices, we can arrive at the following bound:

\begin{proposition}\label{prop:CZ(minrank)} 
     Let $G(V,E)$ represent an n-qubit CZ circuit, and let $t(G)$ be the minimum number of clique flips required to implement $G$. Then, $\minrank_2(G)\leq t(G) \leq \minrank_2(G)+1$, where,
    $$
    \minrank_2(G) = \min \{\mathrm{rank}_{\mathbb{F}_2}(D\oplus G) | D\in \text{diag}(\{0,1\}^n)\}.
    $$
\end{proposition}

We can show stronger bounds when considering specific classes of CZ circuits. One example application of Proposition~\ref{prop:CZ(minrank)} is on CZ circuits represented by random graphs. We can consider Erdos–Renyi random graphs $G(n,1/2)$, where each edge appears with probability $1/2$. We know that for any $g\in G(n,1/2)$, $\minrank(g)> n-\sqrt{2n}$ almost always as $n\to\infty$ \cite{1006.0770}. It also follows that the number of clique flips needed for a graph sampled randomly from $G(n,1/2)$ is almost always about $n$ once $n$ is large enough. I.e., for any $\epsilon, \delta>0$, there exists a $n^*$ s.t. for all $n\geq n^*$, $t(g)\geq (1-\epsilon)n$ with probability $1-\delta$ for randomly sampled $g\in G(n,1/2)$. 

Similar to the drawbacks of \cite{2102.11380}, computation and depth optimization of the circuit becomes a non-trivial task even though the algorithm guarantees the optimal number of clique flips needed. Hence,} we will also survey a simpler method by \cite{2012.09061}, which works by disentangling qubits one by one, as illustrated in Figure~\ref{fig:cliques}~a).

\begin{proposition}\label{prop:CZ(n-1)} \cite{2012.09061}
    Any CZ transformation can be implemented as a circuit using at most n-1 GHZ state injections.
\end{proposition}

\begin{figure}[!htb]
\centering
  \includegraphics[width=\linewidth]{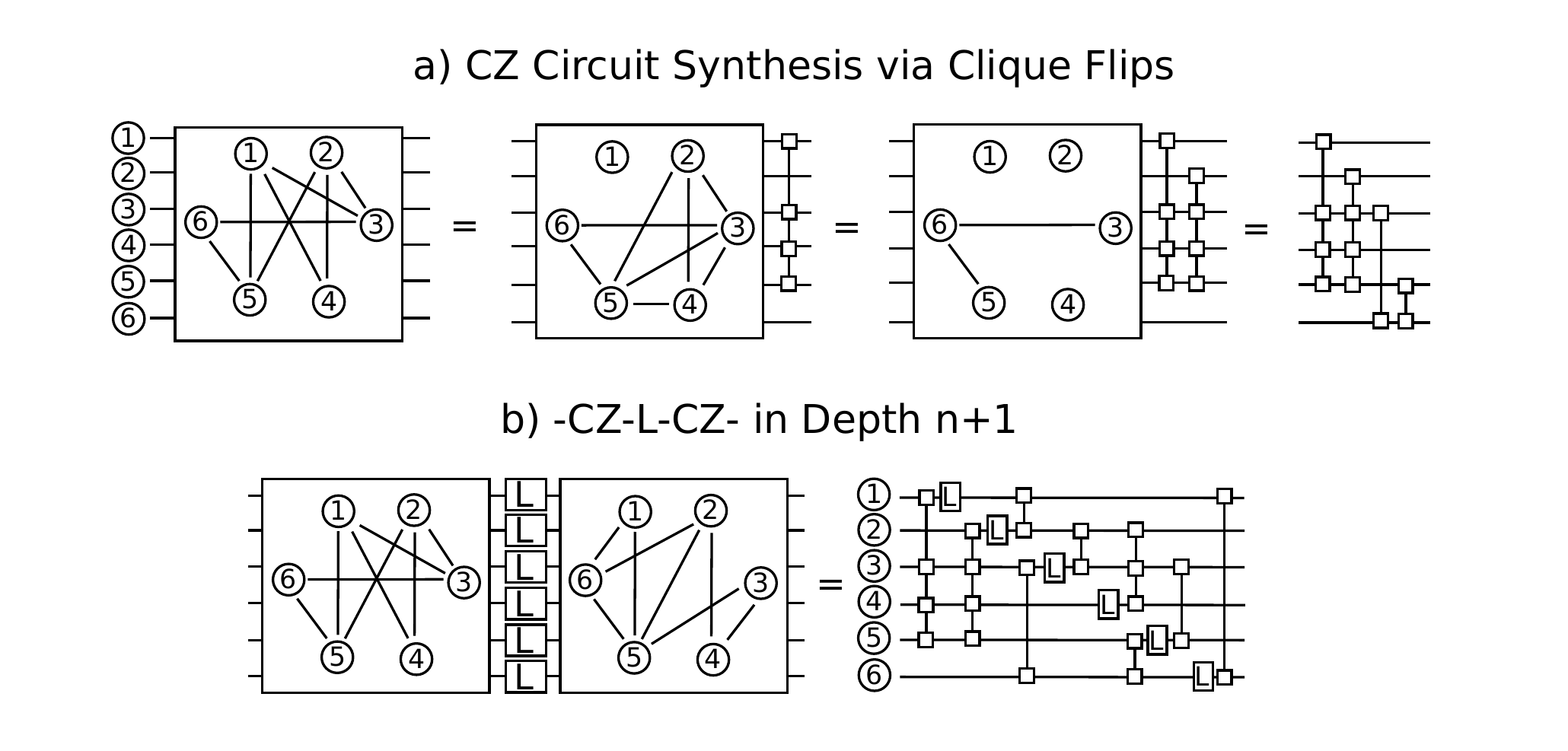}
  \caption{Synthesis of CZ circuits using the clique flip operation defined in Figure~\ref{fig:cliffEg}a). a) Synthesis of a CZ circuit using $\leq n-1$ clique flips from Proposition~\ref{prop:CZ(n-1)}, also shown by \cite{2012.09061}. b) Illustration of the optimization from Proposition~\ref{prop:CZ(n+1)} with two examples of -CZ- layers given by the graphs in the figure. While each CZ layer individually can be synthesized using $\leq n-1$ clique flips following Proposition~\ref{prop:CZ(n-1)}, two such circuits can be slotted together to optimize depth.}\label{fig:cliques}
\end{figure}

\cite{2012.09061}'s construction illustrates a key optimization opportunity: ``stacking''. We observe that as each clique flip disentangles a qubit, successive clique flips contain more isolated vertices, giving ample opportunities for parallelization. If we pick particular orders to disentangle the qubits, we can arrange the circuit in various ``staircases''. Using this idea, we arrive at the following constructions for parallelized circuits implementing -CZ- layers.

\begin{proposition}\label{prop:CZ(n+1)}
    -CZ-L-CZ- can be implemented as a circuit using $2n-2$ clique flips, implementable in GHZ-state-injection depth $n+1$ using a linear GHZ bus.
\end{proposition}

\begin{proposition}\label{prop:CZ(n/2)}
    Any CZ transformation can be implemented as a circuit with GHZ-state-injection depth $\lceil n/2\rceil +1$ in an architecture with two parallel GHZ buses, such as the dual snake architecture.
\end{proposition}

The main idea for Proposition~\ref{prop:CZ(n+1)} is to disentangle qubits in opposite orders so they can be ``stacked'' together, as illustrated in Figure~\ref{fig:cliques}~b); and the main idea for Proposition~\ref{prop:CZ(n/2)} is to cut the graph in two halves, disentangle across the cut, and deal with remaining edges in the two sub-graphs in parallel, as illustrated in Figure~\ref{fig:CZeg}. Proposition~\ref{prop:CZ(n/2)} will later play a central role in our construction for general Clifford gates.

Aside from being useful for Clifford synthesis, the ability to implement arbitrary -CZ- transformations are also closely related to stabilizer state preparation. As shown in \cite{quant-ph/0611214}, all stabilizer states are equivalent to graph states up to an -L- layer, which are a -CZ- layer applied to $\ket{+}^{\otimes n}$. Thus, Proposition~\ref{prop:CZ(n/2)} shows that the dual snake architecture can also prepare stabilizer states in depth $\lceil n/2\rceil +1$.


\subsection{-CX- Transformations}\label{ssec:CX}
Now, we turn our focus to -CX-, denoting an $n$-bit reversible boolean transformation. Let us represent -CX- as an invertible binary matrix $M$, where each column of $M$ corresponds to the output state of a qubit in terms of the inputs. Traditionally, we find a circuit for -CX- by diagonalizing $M$ with column operations (corresponding to CNOT gates) \cite{quant-ph/0701194} which has the best-known depth of $5n$ in LNN \cite{quant-ph/0701194} and $n + o(n)$ in all-to-all \cite{2201.05215}\footnote{The asymptotic optimal depth is $O(\frac{n}{\log{n}})$ \cite{1907.05087}. We do not consider it here due to its impractically large constant overhead.}. With a GHZ bus, we unlock additional abilities to perform up to $n$ row operations simultaneously using fan-out and fan-in gates\footnote{A fan-in is equivalent to fan-out up conjugation by a layer of Hadamard gates.}, which are implementable using one GHZ state.

\begin{proposition}\label{prop:CX(n)}
    Up to a relabeling of qubits, any reversible boolean transformation can be implemented using n fan-outs using a linear GHZ bus.
\end{proposition}

\begin{proof}
    For $i \in [n]$, let $c_i$ be the $i^{th}$ column of $M$ and let $\sigma_i$ be the index of the first non-zero element of $c_i$. With one fan-out, we can add $c_i$ (modulo 2) to all columns $c_j$, $i\neq j$, where $c_j[\sigma_i]$ is non-zero. This reduces $M$ to the permutation matrix given by $[n]\mapsto \{\sigma_1,...,\sigma_n\}$.
\end{proof}

\begin{proposition}\label{prop:CX(2n)}
    Any reversible boolean transformation can be implemented as a circuit with GHZ-state-injection depth 2n-1 using a linear GHZ bus.
\end{proposition}

\begin{proof}
    To ensure the previous algorithm reduced $M$ to a trivial permutation matrix, we need to ensure $\sigma_i = i$ for each $i$. This requires up to one long-range CNOT per fan-out: if $\sigma_i = i$, we are already done; otherwise, we can find $c_j$ where $\sigma_j = i$, and add $c_j$ to $c_i$ using one long-range CNOT. Such $c_j$ always exists for some $j\geq i$; otherwise, $M$ cannot be full rank. 

    We notice that after performing $n-1$ fan-outs controlled by qubits $1,...,n-1$, the last column must have $n = \sigma_n$. Therefore, no additional CNOT is needed for the last fan-out, giving us the depth of $2n-1$ as desired.
\end{proof}

\subsection{Hadamard-Free Clifford Transformations}\label{ssec:Hfree}
Recall that in LNN, a -CZ- layer immediately adjacent to a -CX- layer can be implemented at no additional cost \cite{2210.16195}. This fact also holds in the new model: we give a method for absorbing the -CZ- layer into the -CX- layer. The method below relies on circuit identities relating CZ and CX shown in Figure~\ref{fig:CXCZrules}. For additional clarity, we also give an example of the procedure in Figure~\ref{fig:HFreeEx}.

\begin{figure}[!htb]
\centering
  \includegraphics[width=\linewidth]{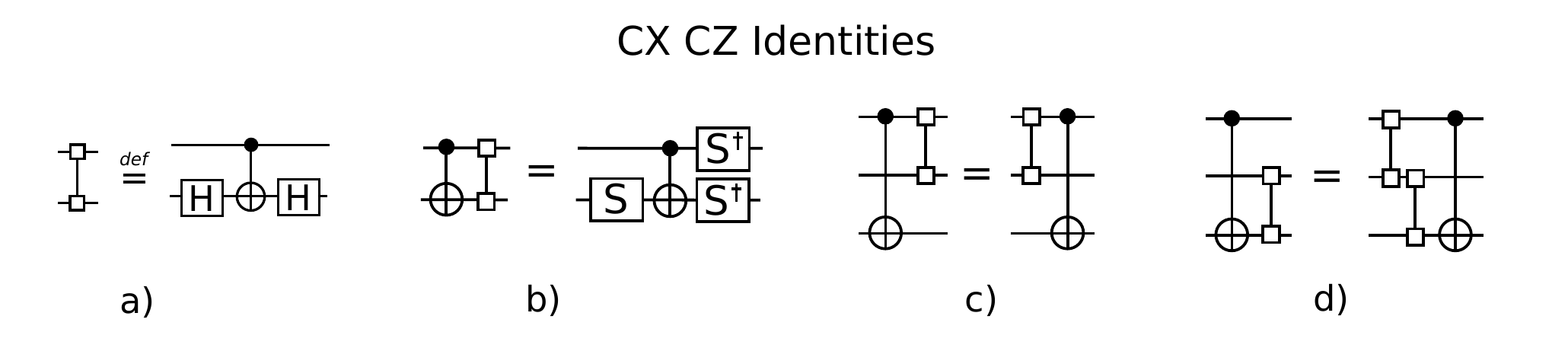}
  \caption{Some identities for commuting a CZ gate through a CNOT gate.}\label{fig:CXCZrules}
  \end{figure}

\begin{figure}[!htb]
\centering
  \includegraphics[width=\linewidth]{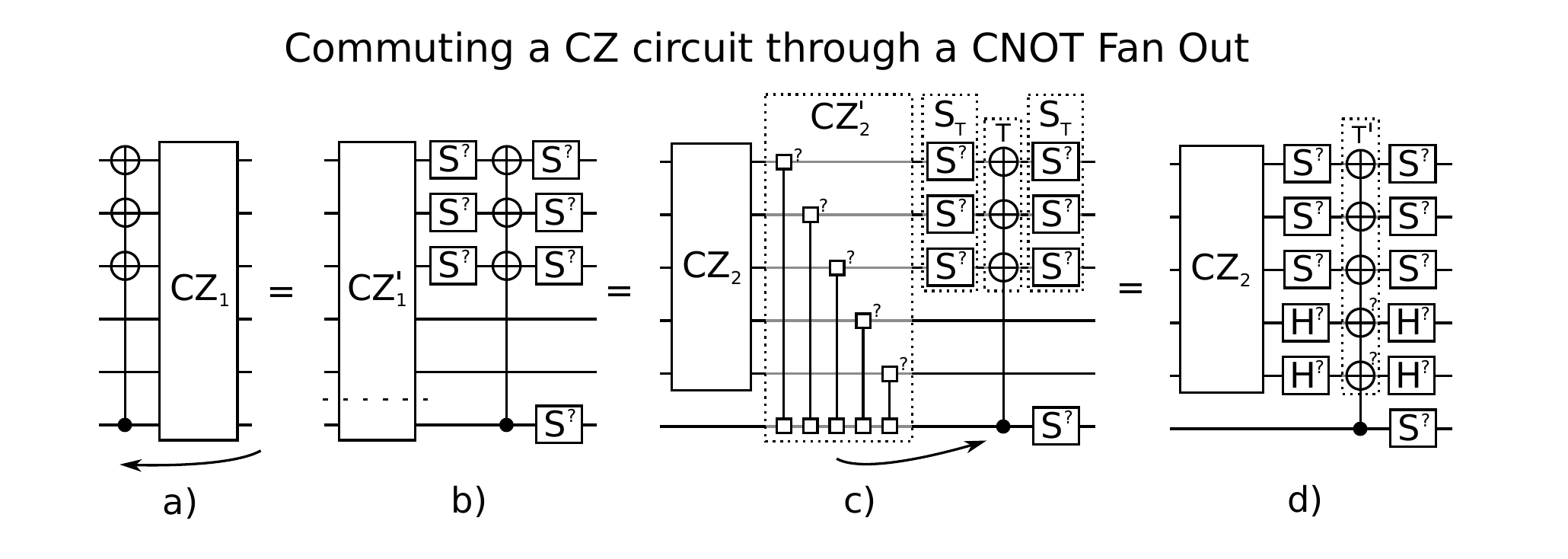}
  \caption{Example of the optimization performed in Proposition~\ref{prop:Hfree(n)}, which absorbs a CZ circuit into a sequence of CNOT fan-outs. a) We leverage the identity in Figure~\ref{fig:CXCZrules}~b) to absorb some gates from $CZ_1$ into some $S$ gates acting on $T \subset [n]$, resulting in $CZ_1'$. b) Gates in $CZ_1'$ touching the control qubit of the fan-out are extracted into $CZ_2'$, with $CZ_2$ left over. c) $CZ_2'$ is absorbed into the fan-out by either removing some $S$ gates or adding additional targets conjugated by $H$ (resulting in $T'\subset [n]$). d) After commuting, the CZ circuit does not touch the target of the fan-out anymore. }\label{fig:HFreeEx}
  \end{figure}

\begin{proposition}\label{prop:Hfree(n)}    
Up to a permutation, any Hadamard-free Clifford transformation can be implemented as a circuit with $n$ fan-outs using a linear GHZ bus.
\end{proposition}

\begin{proof}
We begin with the fact that a Hadamard-free Clifford transformation can be computed as a three-stage computation, -L-CX-CZ-. Let -CX- be written as $n$ fan-outs as in Proposition~\ref{prop:CX(n)}. First, as illustrated in Figure~\ref{fig:HFreeEx}, we can commute a layer of CZ gates through fan-out gates while reducing the width of the -CZ- layer. Given -CZ- circuit $CZ_1$ on qubits $1,...,n$, and a fan-out gate $F$ with control $k$ and targets $T_F\subset [n] - \{k\}$, we will describe in three steps how this commutation is achieved:
\begin{enumerate}
    \item Commute $CZ_1$ through the $F$ using well known circuit identities given in Figure~\ref{fig:CXCZrules}. We have $F \cdot CZ_1 = CZ_1'  \cdot S_T \cdot F \cdot S_{T\cup n}$, where $S_T$ denotes a layer of single qubit phase gates on qubits in $T$.
    \item Partition $CZ_1'=CZ_2\sqcup CZ_2'$, where $CZ_2 = \{\cz(i,j) | \cz(i,j)\in CZ_1', i,j\neq n \}$, and $CZ_2' = CZ_1' - CZ_2$. That is, $CZ_2'$ consists of all CZ gates on qubit $n$ and $CZ_2$ consists of all other CZ gates.
    \item There are two cases for gates CZ$(i,n)\in CZ_2'$.
    \begin{enumerate}
        \item $i\in T$. In this case, there exists CX$(n,i)\in F$; the CZ gate can be implemented using phase gates. 
        \item $i\notin T$. Rewrite CZ$(i,n)$ as $H_i$ CX$(n,i)H_i$; we notice this CX gate can be merged with $F$ to obtain a new fan-out $F'$ with the same control and targets $T' = T\cup \{i\}$, up to a layer of Hadamard gates.
    \end{enumerate}
\end{enumerate}

It follows that $F\cdot CZ_1 = CZ_2 \cdot F'$, where $CZ_2$ does not contain any CZ gates that act on the control of $F$, and $F'$ is a fan-out gate with the same control as $F$ and (possibly) more targets, up to conjugation by single qubit phase gates and Hadamard gates. $F$ and $F'$ can both be implemented using one GHZ state injection, up to some irrelevant single qubit gates.

We can repeatedly commute the CZ circuit to obtain $CZ_2,...,CZ_{n}$. Since the $n$ fan-out gates given by Proposition~\ref{prop:CX(n)} have distinct controls, each time we pass by a fan-out layer, the width of the CZ circuit decreases by 1. Hence $CZ_n = I$; we have implemented -CZ- inside -CX- with no additional cost, as desired.
\end{proof}

A simple corollary follows that a Hadamard-free operation can be implemented in depth $2n-1$, since we can still easily commute CZ gates through the additional CNOT layers. Furthermore, we also notice that if $CZ_1$ does not act on qubits $Q_n = \{i_1,...,i_m\}$, and $k\notin Q_n$ where $k$ is the control of $F$, then $CZ_2$ does not act on $Q_n$, and $F'$ does not add additional targets to qubits in $Q_n$. For all intents and purposes, the commutation rule given above leaves gates on $Q_n$ unchanged. 

Given this observation, in fact, a second -CZ- layer can also be implemented at no additional cost when implementing a -CX- circuit exactly using a GHZ bus. We will show this construction in the next subsection.

\subsection{Clifford Transformations}\label{ssec:Cliff}
Putting everything together, we arrive at the main result for the linear GHZ bus model. An example of this algorithm is given in Figure~\ref{fig:cliffEg}.

\begin{corollary}\label{cor:cliff(2n)}
    Any Clifford transformation can be implemented as a circuit with GHZ-state-injection depth $2n+1$ using a linear GHZ bus.
\end{corollary}
\begin{proof}
    Let us first consider an alternative decomposition of Clifford operations, -L-CZ-CX-XCX-L-, where -XCX- denotes a layer of X-controlled-NOT-gates: $\xcx:=H^{\otimes 2}\cdot  \cz \cdot H^{\otimes 2}$. We can obtain this decomposition by commuting the full layer of Hadamard gates through the second -CZ- layer in the scheme given by \cite{2003.09412}. 
    
    First, we can synthesize the -CX- layer using Proposition~\ref{prop:CX(2n)}, where odd layers $2i-1$ contain one CNOT gate on qubits $i,j$ where $j> i$, and even layers $2i$ contain a fan-out controlled by qubit $i$. From the left, we can push in a -CZ- circuit using techniques described in Proposition~\ref{prop:Hfree(n)}. From the right, we can first decompose the -XCX- circuit as an upside-down staircase using techniques described in Proposition~\ref{prop:CZ(n+1)}, and commute them to stack on top of the CNOTs in the odd layers. This is always possible since XCX gates commute with the target of a CNOT gate, and the controls of the fan-outs are in descending order. Overall, the depth is increased by 2. 
\end{proof}

\begin{figure}[h]\label{fig:cliffEg}
    \centering
    \includegraphics[width=\linewidth]{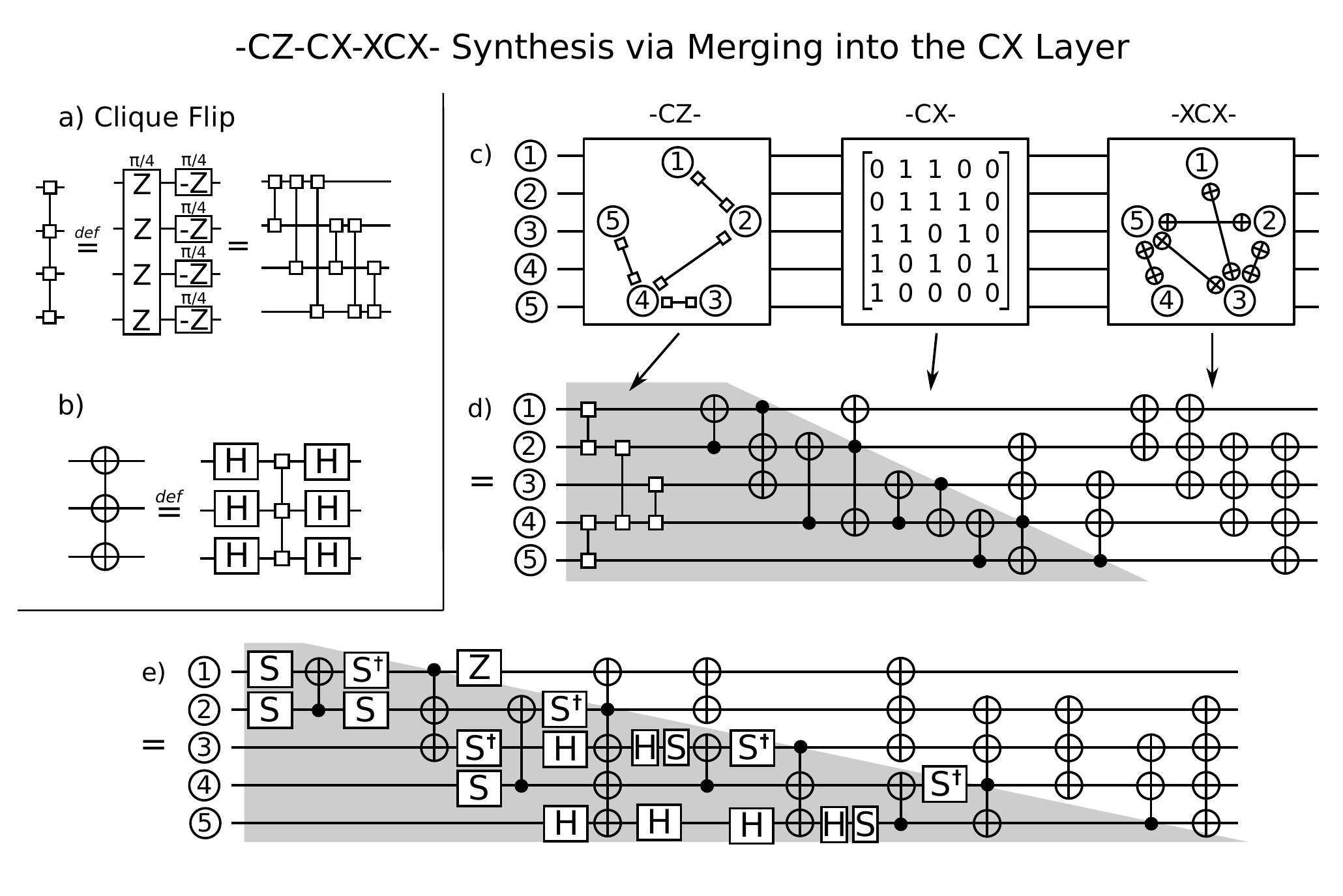}
    \caption{ a) The `clique flip' operation is a particular Pauli rotation with local corrections and can be shown to implement CZ on all pairs of the involved qubits. This operation is also considered by \cite{1707.06356, 2012.09061}.  b) XCX gates have an analog of the clique flip operation.  c) A Clifford circuit is equivalent to -CZ-CX-XCX- up to some local gates.  d) The CX circuit is synthesized using Proposition~\ref{prop:CX(2n)} with the controls on qubits with descending labels, and the -XCX- circuit can be built into an upward-facing triangle using Proposition~\ref{prop:CZ(n-1)}. We do not care about the CZ synthesis since in e) we use the method from Proposition~\ref{prop:Hfree(n)} to absorb the CZ circuit into the downward-facing part of the CX circuit. Since XCX clique flips commute with CNOT fan-out targets, we commute them through and stack on top of the CNOTs in the odd layers.  \label{fig:cliffEg}}
   
\end{figure}

We also present a similar result in the dual snake model:

\begin{corollary}\label{cor:cliff(1.5n)}
    Any Clifford transformation can be implemented as a circuit with GHZ-state-injection depth $\lceil \frac{3}{2}\rceil n + O(\sqrt{n})$ in a square-lattice architecture supporting the dual snake layout.
\end{corollary}
\begin{proof}
It is sufficient to be able to implement a Hadamard-free transformation and a CZ transformation \cite{2003.09412}.  Up to a permutation, a Hadamard-free Clifford transformation can be implemented in depth $n$ by Proposition~\ref{prop:Hfree(n)} and a -CZ- circuit can be implemented in depth $\lceil\frac12 \rceil n +O(1)$ by Proposition~\ref{prop:CZ(n/2)}. Finally, a permutation can be implemented in depth $O(\sqrt{n})$ on a square lattice \cite{10.1145/12130.12156}, where adjacent horizontal and vertical SWAPs can be implemented efficiently by using the ancilla qubits otherwise dedicated to the GHZ bus. 
\end{proof}

\section*{Author Contribution Statement}

The authors worked together on developing the appropriate model of computation and iterating on synthesis approaches. WY contributed the main synthesis theorems discussed in Section~\ref{sec:LayeredSynth} and Appendix~\ref{app:CZsynth}. PR contributed the state injection derivations underpinning the model highlighted in Appendix~\ref{app:GHZ_preparation_injection}, as well as the counting arguments in Appendix~\ref{app:lower_bounds}. Both authors are supported by IBM Quantum.

\bibliography{CATbib}

\appendix

\section{Lower Bounds for Clifford Circuits}\label{app:lower_bounds}
Here we present some simple counting arguments.

\begin{proposition} Any sequence of $m$ many $n$-qubit Pauli rotations that implements an arbitrary element of the Clifford group will require $m \geq n$.
\end{proposition}
\begin{proof} Recall that $\log_2 |\mathcal{C}_n| \geq 2n^2 +n$ bits are required to specify an element of the Clifford group \cite{2003.09412}. Each Pauli rotation on $n$ qubits encodes $2n+1$ bits, so at least $(2n^2+n) /(2n+1) = n$  are required.
\end{proof}

More generally, since we allow Pauli rotations, fan-out gates following a GHZ state injection, as well as arbitrary single-qubit Clifford gates, we need to be more careful when deriving general lower bounds for our model. 

\begin{proposition}\label{prop:lowerbound}
Any circuit consisting $m$ layers of parallelizable Pauli rotation and fan out gates that implements an arbitrary element of the Clifford group will require $m \geq 0.648n-2$.
\end{proposition}
\begin{proof} 
First, we observe that since the Pauli matrices are normalized by the Clifford group, we can commute all single-qubit Clifford gates to the beginning of the circuit. This may alter the elements of the Pauli rotation, or change the controls and the targets of the fan-out gate to arbitrary Paulis (instead of the Z- and X-targets). Let's call them conjugated fan-out gates. Additionally, since $\frac\pi2$ Pauli rotations are local, we may ignore the sign of a Pauli rotation. WLOG, we can consider a canonical form where circuits consist of one layer of single-qubit Clifford gates followed by $m$ layers of parallelizable Pauli rotations and conjugated fan-out gates, and we are interested in finding a lower-bound on $m$ such that any $n$-qubit Clifford operation requires at least $m$ such layers. Similarly as above, we proceed by finding an upper bound on the number of bits required to specify one such layer of global gates. We will furthermore assume all the gates are conjugated fan-out gates since they require strictly more information to specify compared to a Pauli rotation acting on the same qubits.

Suppose the layer of the global gate acts non-trivially on $k$ qubits. Then there are $3^k$ ways to specify the non-identity matrices and $k-1$ locations where the string can be uniquely split. It remains to specify a control qubit for each segment of the chain, for which there are at most $l$ choices for a segment containing $l$ non-identity elements. Overall, with $s$ segments, there are $l_1\times...\times l_s$ choices where $l_1+...+l_s = k$, which is upper-bounded by $2^{\frac{k}{2}}$ when $s=k/2$. the total number of possibilities is at most:

\begin{align}
\sum_{k=0}^{n}\left(  3^{k} \cdot  2^{k-1} \cdot 2^{\frac{k}{2}} \right) = \frac{(6\sqrt{2})^{n+1} - 1}{6\sqrt{2}-1}  \leq \frac{6\sqrt{2}}{6\sqrt{2}-1}(6\sqrt{2})^n
\end{align}

Finally, there are $24^n$ choices for the layer of single qubit gates. Since there are $2^{2n^2 +n}$ Clifford operations, we will require at least
\begin{align}
    m\geq \frac{2n^2+n - n\log(24)}{n\log(6\sqrt{2}) - \log(\frac{6\sqrt{2}}{6\sqrt{2}-1})}\geq 0.648 n - 2 \label{eq:lowerbound}
\end{align}


\end{proof}

\section{Derivations, Proofs, and Examples}\label{app:derivations_examples}

In this section, we give additional details about our model and detailed proofs of the CZ synthesis schemes with some illustrative examples.

\subsection{GHZ Preparation and Injection}\label{app:GHZ_preparation_injection}

In Section~\ref{sec:Model} and Figure~\ref{fig:LNN_GHZ_Bus} we gave an overview of the capabilities of the architectures considered in this manuscript. In this section, we present additional details as to how these capabilities are achieved.

A quantum circuit for synthesizing GHZ states on a GHZ bus is presented in Figure~\ref{fig:GHZ_prep}~a). This constant-depth circuit requires two CNOT layers and one measurement layer to execute in an architecture where CNOT gates are native. However, in a surface code architecture, there are more direct ways of implementing long-range GHZ states: a large rectangular ancilla patch storing a single qubit of data can be prepared in a single code cycle \cite{1808.02892}. Then, the circuit in the figure merely presents what is happening `at a logical level' and highlights that the ability to synthesize GHZ states primarily stems from the ability to perform mid-circuit measurements and apply Pauli corrections. Since $Z\otimes Z$ measurements are native operations in a lattice-surgery architecture, a similar interleaving trick as in \cite{2110.11493} can be applied to synthesize two GHZ states across two intersecting GHZ buses. 

\begin{figure}[!htb]
\centering
 \includegraphics[width=\linewidth]{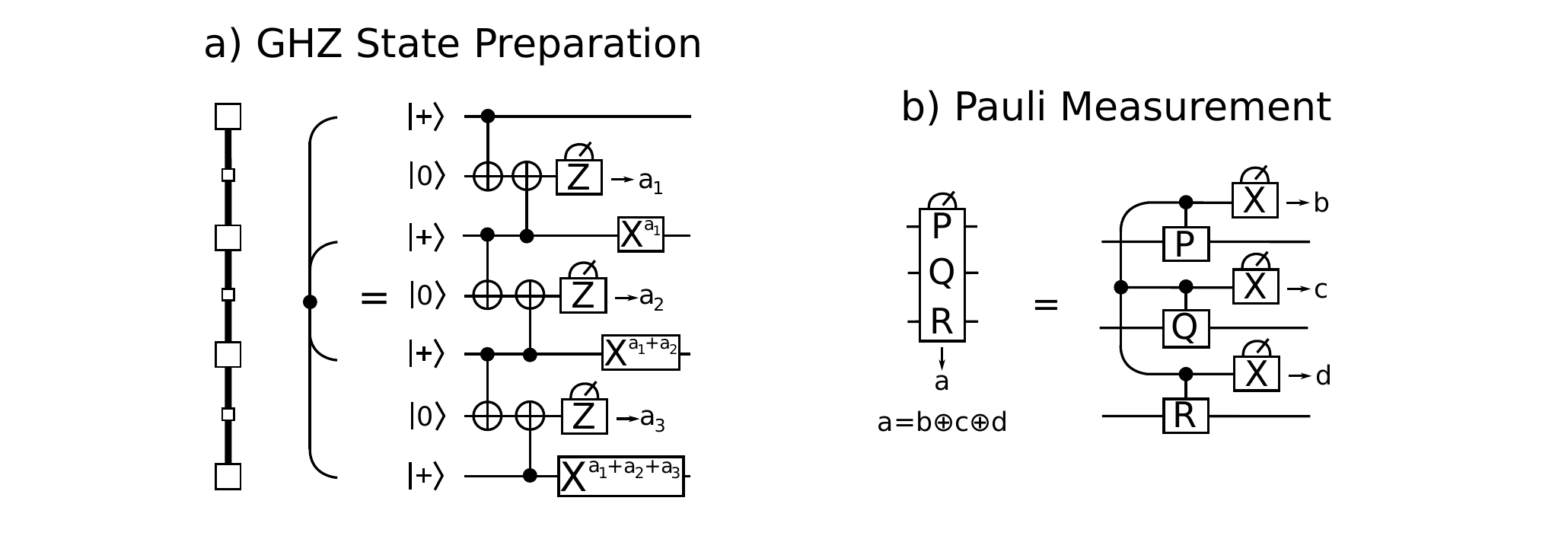}
  \caption{a) Preparation of a GHZ state on the ancillae of the GHZ bus using parity check qubits, as well as two CNOT layers and a mid-circuit measurement. b) GHZ state injection implementation of another primitive not leveraged in our work: measurement of an $n$-qubit Pauli observable.}\label{fig:GHZ_prep}
  \end{figure}

In our circuit constructions, we only consider the Pauli rotation and CNOT fan-out gates enabled by GHZ state injection. Another primitive operation enabled by GHZ state injection is multi-qubit Pauli measurement, as shown in Figure~\ref{fig:GHZ_prep}~b). The circuit shown can be seen as a surface code agnostic representation of ancilla-based measurement which forms a central tool in the architecture presented in \cite{1808.02892}. Of course, the cost model underpinning our constructions crucially relies on CNOT operations having roughly the same cost as GHZ state injections which is only true in surface codes in the first place. Nonetheless, it is interesting to show how these operations may be implemented in a more general architecture. Figure~\ref{fig:derivations} gives derivations using ZX calculus for the three primitive operations: fan-out, Pauli rotation, and Pauli measurement.

\begin{figure}[!htb]
\centering
  \includegraphics[width=\linewidth]{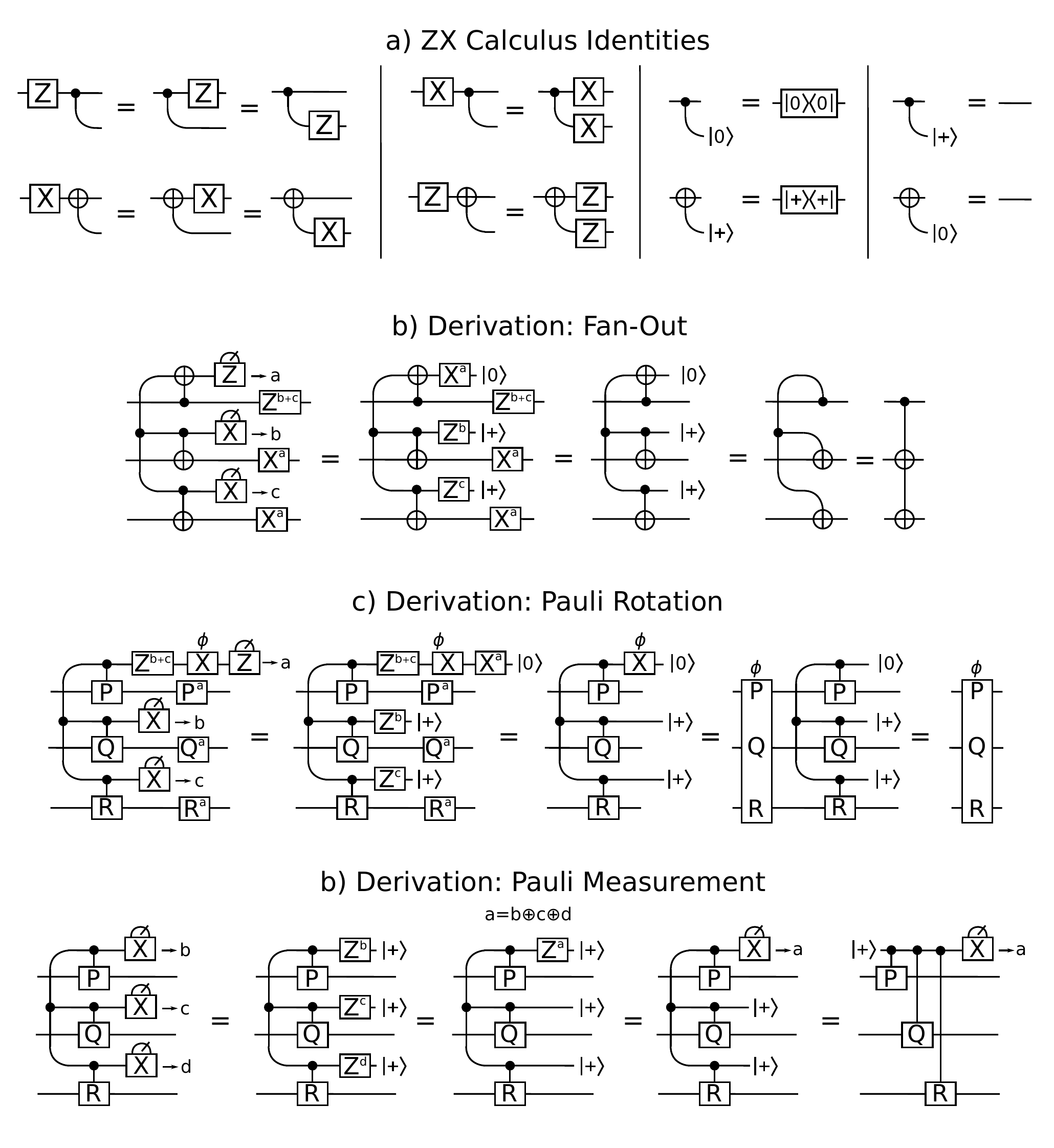}
  \caption{Derivation of GHZ state injection circuits using ZX calculus \cite{0906.4725}.}\label{fig:derivations}
  \end{figure}

Since each of these primitive operations demands one GHZ state to implement, it stands to reason that these should all roughly have the same power. Certainly, it is easy to see how to prepare a single GHZ state using CNOT fan-out. We can also prepare a GHZ state by applying $\exp(i\frac{\pi}{4} Y^{\otimes n})$ to $\ket{0^n}$, or by measuring the $X^{\otimes n}$ observable on $\ket{0^{n}}$ and applying a Pauli correction. Interconversion of the operations is less simple, and circuits achieving these are given in Figure~\ref{fig:interconversion}.  We find that to transform one of these operations into any of the other two, an additional ancilla qubit is required. This makes sense for Pauli measurements since they require an additional degree of freedom to be measured in order to avoid damaging the coherence of the input state.  However, the smallest circuit without an extra ancilla for implementing fan-out using Pauli rotations requires two clique flips: one on all the qubits, and on all but the target. Even with the additional ancilla, the synthesis of fan-out gates demands an additional CNOT gate. But even with these limitations, there is plenty of evidence that these three circuit primitives have roughly the same capabilities even up to constant factors.

\begin{figure}[!htb]
\centering
  \includegraphics[width=\linewidth]{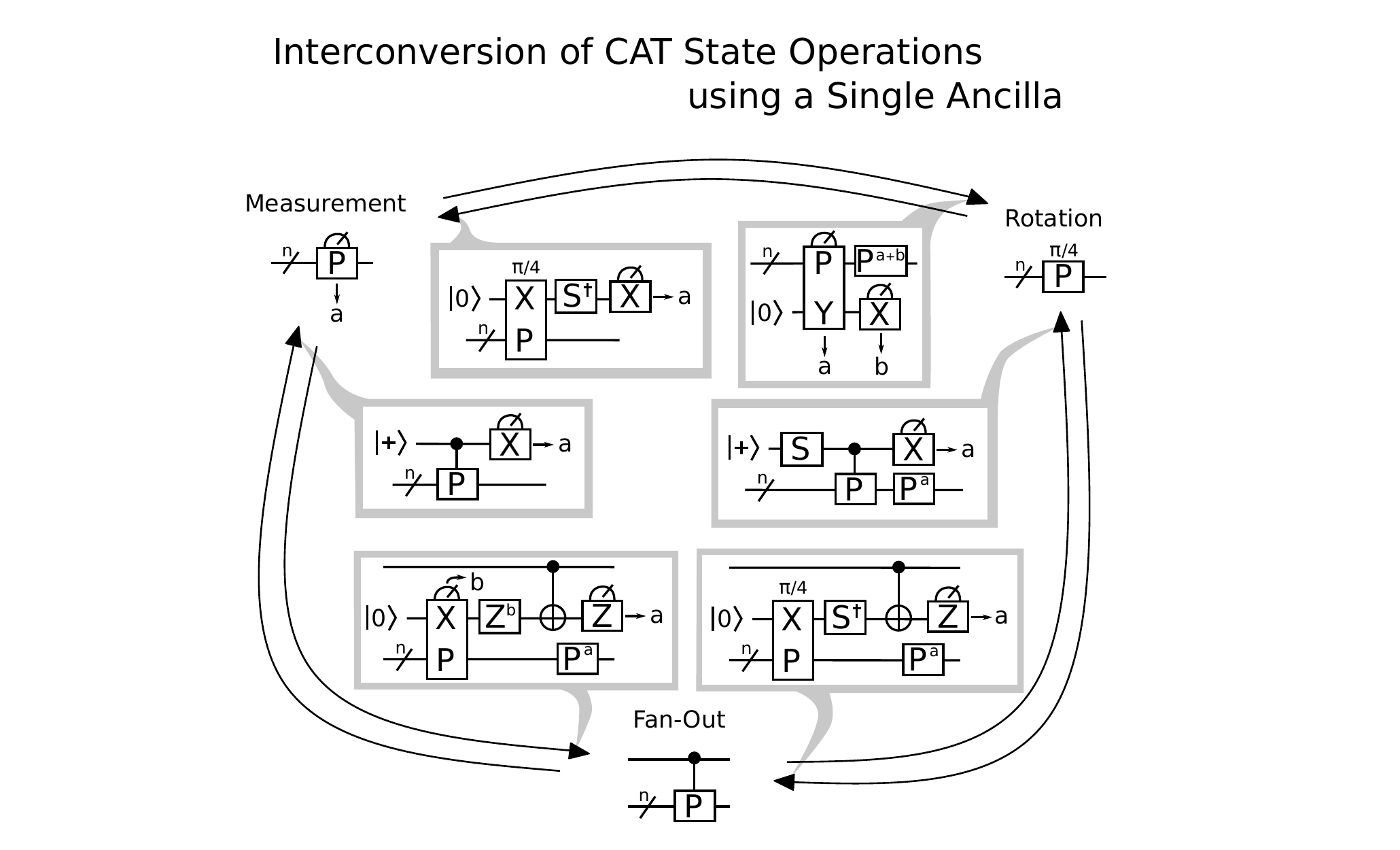}
  \caption{Circuits for interconversion of the three GHZ state enabled $n$-qubit gates considered in this paper: Pauli measurement, Pauli rotation, and fan-out. All of these conversions require an ancilla qubit, and synthesis of fan-out requires an additional CNOT gate. But otherwise, this is evidence that these three operations have roughly the same power.}\label{fig:interconversion}
  \end{figure}

\subsection{CZ Synthesis Proofs}\label{app:CZsynth}

Here we give the proofs underpinning Propositions~\ref{prop:CZ(n+1)}, \ref{prop:CZ(n+1)}, and \ref{prop:CZ(n/2)} for the synthesis of CZ transformations. We represent an $n$-qubit CZ transformation as a graph $G(V,E)$, where each vertex in $V$ corresponds to a qubit, and each edge $(v_1, v_2) \in E$ indicates a CZ gate between qubits $v_1,v_2$. Since all CZ gates commute and are self-inverse, concatenating two CZ transformations $G_1(V,E_1), G_2(V,E_2)$ gives a new CZ transformation $G_3(V,E_3)$ where $E_3$ is the symmetric difference of $E_1, E_2$. With a slight abuse of notation, let us also denote $G$ as the adjacency matrix. Then, $G_3 = G_1\oplus G_2$. We assume WLOG that all CZ transformations share a common set of vertices. Note in this representation, the application of a clique flip corresponds to the concatenation of a complete graph on a subset of vertices; in other words, it `flips' all the edges corresponding to a clique, hence the name.
{{\color{blue}}
\begin{proof}[Proof of Proposition~\ref{prop:CZ(minrank)}]
    If $t(G) = 1$, then $G$ contains exactly one clique, and $\minrank(G) = t(G) = 1$. In particular, the minrank is achieved by choosing $D = I$. 
    
    First, let's show $\minrank(G)\leq t(G)$. Suppose $t(G) = m$ for some $m>1$, $G = K_1\oplus...\oplus K_m$ where $t(K_i) = 1$ for each $i\in[m]$. Since minrank is sub-additive, $\minrank_2(G) =  \minrank_2(\bigoplus_{i=1}^m K_i) \leq \sum_{i=1}^m \minrank_2(K_i) = t(G)$.

    Then, we'll show $t(G)\leq \minrank(G)+1$. Suppose $\minrank_2(G) = r$; then, there exists $G^* = D^* \oplus G$ where $\mathrm{rank}_{\mathbb{F}_2}(G^*) = r$. Since $G^*$ is symmetric, we can use Lempel's factorization \cite{10.1137/0209059} to find an $n\times r'$ dimensional factor $F$ such that $G^* = FF^T$, where $r' = r+1$ if $G^* = G$ and $r' = r$ otherwise. Let $f_i$ be the $i^{th}$ column of $F$, we can rewrite $G^* = \bigoplus_{i=1}^{r'}f_if_i^T$. Note $f_if_i^T = diag(f_i)\oplus K_i$, where $K_i$ is a complete graph on vertices $\{j|F_{ij}=1\}$. Therefore, $t(G) = t(K_1\oplus...\oplus K_{r'}) \leq \minrank_2(G)+1$.
\end{proof}
}

For completeness, we also present the construction from \cite{2012.09061} for synthesizing a CZ transformation using $n-1$ clique flips, since this idea plays a central role in the proof of our main result. The method involves iteratively disentangling single qubits from the transformation. An example is shown in Figure~\ref{fig:cliques}~a).

\begin{proof}[Proof of Proposition~\ref{prop:CZ(n-1)}]
    To synthesize $G(V,E)$, we can find $G_1,...,G_m$ s.t. $G_1\oplus...\oplus G_m = G$, where each $G_i$ consists of a clique, implementable using one clique flip via a GHZ state injection. There is a simple algorithm to find these cliques. For each $i\in [n]$, let $S_i = [\bigoplus_{j=1}^{i-1} G_j]\oplus G$ be the graph left over after applying all $G_j$ up to $i-1$, and set $G_i$ to be the complete graph on $N_{S_i}(v_i)\cup \{v_i\}$\footnote{$N_G(v)$ refers to the $v$'s neighbors in $G$; i.e., $N_{G(V,E)}(v):= \{u|u\in V, (u,v)\in E\}$}; that is, a clique on $v_i$ and its neighbors in $S_i$. Notice that $v_i$ becomes an isolated vertex in $S_{i+1} = G_i \oplus S_i$ since the concatenation of $G_i$ will cancel out any edges from $v_i$ in $S_i$. It follows that $S_n = [\bigoplus_{i=1}^{n-1} G_i]\oplus G$ have only isolated vertices; hence $G = G_1\oplus...\oplus G_{n-1}$, as desired.
\end{proof}

We observe that the triangular structure of the resulting circuit can be exploited to execute two CZ transformations simultaneously, thereby allowing us to synthesize a -CZ-L-CZ- construction using depth $n+1$ despite requiring $2n-2$ clique flips. An example is shown in Figure~\ref{fig:cliques}~b). Since the supports of these clique flips never overlap, using the more powerful dual snake architecture is not necessary: a single GHZ bus suffices.

\begin{proof}[Proof of Proposition~\ref{prop:CZ(n+1)}]
    First, let the vertices be ordered from $1,...$ to $n$. For the first CZ transformation, pick $G^1 = G^1_1\oplus ...\oplus G^1_{n-1}$ as in Proposition~\ref{prop:CZ(n-1)} where $G^1_i$ accounts for the CZ gates related to $v_i$. Then, $v_1,...,v_{i-1}$ are isolated vertices in $G^1_i$; hence, the corresponding clique flip does not act on qubits $1,...,i-1$. This will be our first staircase.
    
    For the second CZ transformation, let us fix the vertices in reverse: pick $G^2 = G^2_1\oplus...\oplus G^2_{n-1}$ where $G^2_i$ accounts for the CZ gates related to $v_{n-i+1}$. Here, the clique flip corresponding to $G^2_i$ only acts on $v_1,...,v_{n-i}$. This will be our second, upside-down staircase.

    It follows that for $i = 3,..., n-2$, $G^1_i$ and $G^2_{n-i+1}$ can be implemented in parallel, giving us a total depth of $n - 3 +4 = n +1$.
\end{proof}

Finally, we give a construction that exploits the power of the dual snake model to implement two clique flips simultaneously even if their supports overlap, provided they act on disjoint sets of qubits. This capability synthesizes a CZ transformation using GHZ-state-injection depth $\lceil n/2\rceil +1$. This result immediately applies our stated bound for stabilizer state preparation and is a key ingredient in the construction of Clifford gate synthesis.

The basic idea is to cut the graph into two halves. Once the two halves have been separated, each can be synthesized using Proposition~\ref{prop:CZ(n-1)}. The strategy to separate the graph is to deal with the first and last qubits simultaneously, then the second and second-to-last qubits, and so on, observing that the separation can either be performed using a single clique flip or two non-overlapping ones. An example of this construction is given in Figure~\ref{fig:CZeg}.

\begin{proof}[Proof of Proposition~\ref{prop:CZ(n/2)}]
    Let us find a bipartition of the vertices $V = V_l\sqcup V_r$, where $V_l = \{v_1,...,v_{\lceil n/2\rceil}\}$ and $V_r = \{v_{\lceil n/2\rceil+1},...,v_n\}$. This bipartition defines a cut on $G$.
    
    We will first address the CZ gates that cross the cut. For $i = 1,...,\lceil n/2\rceil$, let $S_i = [\bigoplus_{j=1}^{i-1} G^c_i ]\oplus G$ and let $C_i$ be the edges that cross the cut in $S_i$, where $G^c_i$ is constructed as:
    \begin{enumerate}
        \item Two cliques, one on $V^l_i = N_{C_i}(v_i)\cup \{v_i\}$, and one on $V^r_i = N_{C_{i}}(v_{n-i+1}) \cup \{v_{n-i+1}\}$, if $(v_i, v_{n-i+1})\notin C_i$
        \item One clique on vertices $V_i = V^l_i\cup V^r_i$, if $(v_i, v_{n-i+1})\in C_i$.    \end{enumerate}
        
     In either cases, we observe that 1), $C_{i+1}$ does not contain any edges that have endpoints $v_i, v_{n-i+1}$, and 2), vertices $v_1,...,v_i$ and $v_n-i+1,... v_n$ are isolated in $G^c_{i+1}$. As a result of 1), $C_{\lceil n/2\rceil+1}$ is empty; it remains to deal with the edges contained in $V_l$ and $V_r$. Given 2), we can implement the clique flips for the two disconnected sub-graphs in parallel using the staircases given in Proposition~\ref{prop:CZ(n+1)}. The parallelization increases the depth by at most 1.

\end{proof}

\begin{figure}[!htb]
\centering
  \includegraphics[width=\linewidth]{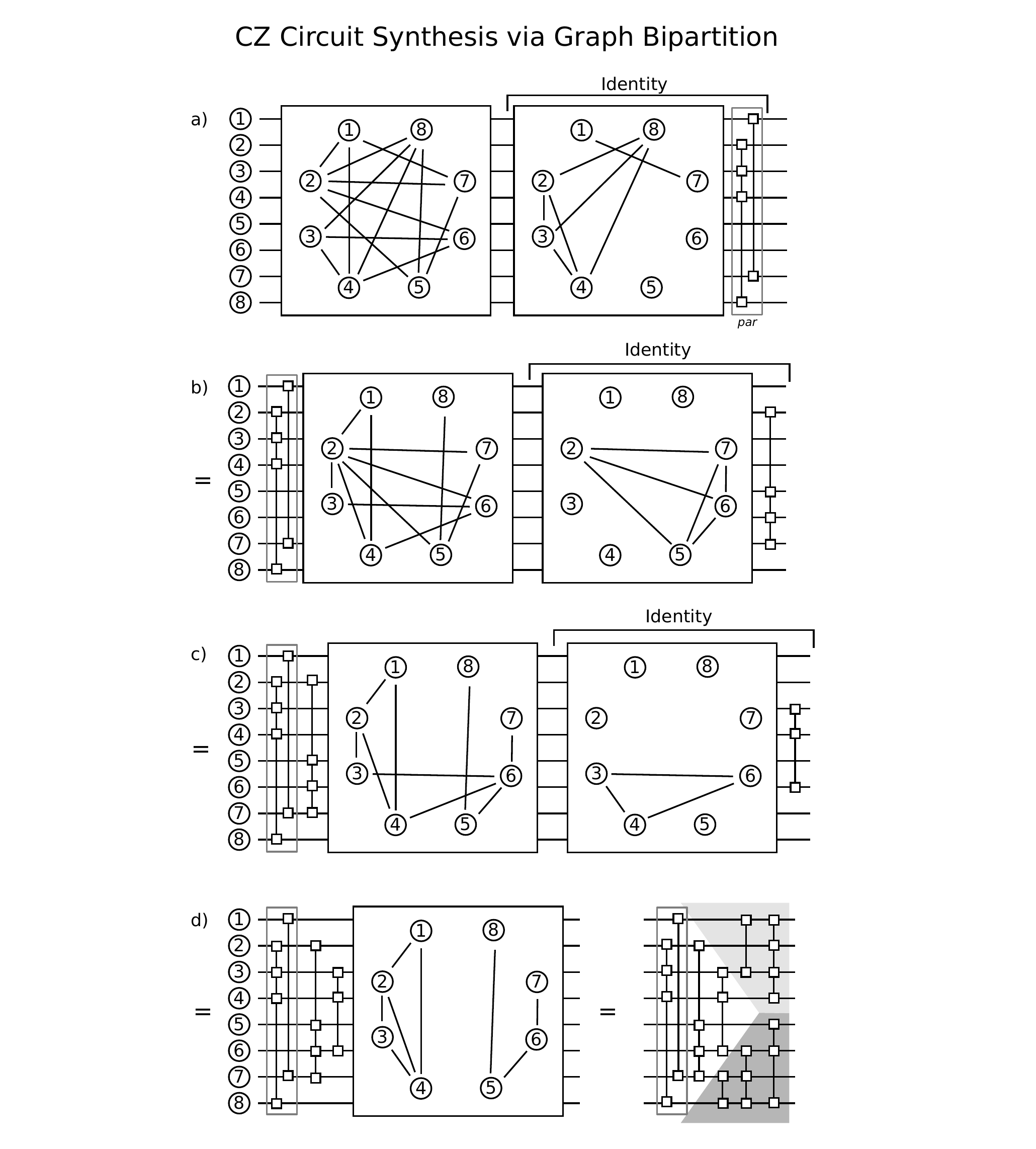}
    \caption{Example of the synthesis algorithm from Proposition~\ref{prop:CZ(n/2)}. The method first servers all edges between the groups $V_l = \{1,2,3,4\}$ and $V_r = \{5,6,7,8\}$. a) Considering qubits 1 and 8 which are not connected (case 1),  we can eliminate the edges across the cut using two clique flips. While these clique flips cannot be parallelized in a model with one GHZ bus, they can be with two GHZ buses. b) Considering qubits 2 and 7 which are connected (case 2), we can eliminate the edges with one clique flip. b) Similarly 3 and 6 correspond to case 2. We have removed all edges across the bipartition using $\lceil n/2\rceil$ layers. c) Finally, the two remaining graphs on $V_l$ and $V_r$ can be synthesized using Proposition~\ref{prop:CZ(n-1)}, and due to the triangular structure of the circuits, this requires only one additional layer. \label{fig:CZeg} }
  \end{figure}

\end{document}